\documentclass[letterpaper,11pt]{article}

\usepackage{times}

\usepackage{fullpage}

\usepackage{amsmath}
\usepackage{amsfonts}
\usepackage{amssymb}
\usepackage{amsthm}
\usepackage{amsbsy}
\usepackage{graphicx}
\usepackage{subcaption}
\usepackage{verbatim}
\usepackage{url}
\usepackage{bbm}
\usepackage{multirow}
\usepackage{bm}

\newtheorem{theorem}{Theorem}

\theoremstyle{definition}

\newcommand{\bb}{\mathbf{b}}

\newcommand{\vv}{\mathbf{v}}
\newcommand{\R}{\mathbb{R}}

\newcommand{\cc}{\mathbf{c}}

\newcommand{\yy}{\mathbf{y}}

\newcommand{\LW}{\mbox{\normalfont LW}}
\newcommand{\OPT}{\mbox{\normalfont OPT}}
\newcommand{\lpoa}{\mbox{\normalfont LPoA}}
\newcommand{\lpos}{\mbox{\normalfont LPoS}}

\newcommand{\calG}{\mathcal{G}}

\newcommand{\eq}{\mbox{eq}}
\newcommand{\GSP}{\mbox{GSP}}
\newcommand{\VCG}{\mbox{VCG}}
\newcommand{\EGFP}{\mbox{EGFP}}

\newcommand{\remove}[1]{}
\begin{document}

\title{\bf\ A note on the efficiency of position mechanisms \\ with budget constraints\thanks{This work has been partially supported by a PhD scholarship from the Onassis Foundation, and by the European Research Council (ERC) under grant number 639945 (ACCORD).}}

%\author{\name Alexandros A. Voudouris \email \normalfont alexandros.voudouris@cs.ox.ac.uk \\
%       \addr University of Oxford}

\author{Alexandros A. Voudouris \\ University of Oxford \\ alexandros.voudouris@cs.ox.ac.uk} 
\date{}

\maketitle     
       
\begin{abstract}
We study the social efficiency of several well-known mechanisms for the allocation of a set of available (advertising) positions to a set of competing budget-constrained users (advertisers). Specifically, we focus on the Generalized Second Price auction (GSP), the Vickrey-Clarke-Groves mechanism (VCG) and the Expressive Generalized First Price auction (EGFP). Using the liquid welfare as our efficiency benchmark, we prove a tight bound of $2$ on the liquid price of anarchy and stability of these mechanisms for pure Nash equilibria.
\end{abstract}

\section{Introduction}
Position mechanisms have been widely used for the allocation of advertising positions (with different click-through rates) when keywords are queried in search engines. Such mechanisms auction off the available positions to the interested advertisers, who in turn compete with each other by submitting bids, expressing how much they value the available advertising positions (per user click). 

There have been numerous papers analyzing the properties of position mechanisms. Edelman {\em et al.}~\cite{EOS07} (see also \cite{V07}) studied the {\em generalized first price auction} (GFP) as well as the {\em generalized second price auction} (GSP). According to these mechanisms, the advertisers are sorted in terms of the scalar bids that they submit, and each of them pays her own bid or the next highest bid, respectively. The definition of the mechanisms allow the advertisers to strategize over their bids and engage as players into a strategic game. Edelman {\em et al.}~\cite{EOS07} proved that the games induced by GFP are not guaranteed to have pure Nash equilibria, while the games induced by GSP always have socially efficient pure Nash equilibria with respect to the {\em social welfare} benchmark (the total value of the players for the positions they are given); consequently, the {\em price of stability}~\cite{ADK+08} of GSP is equal to $1$. Caragiannis {\em et al.}~\cite{CKK+15} (see also \cite{ST13}) focused on worst-case equilibria and proved several bounds on the {\em price of anarchy}~\cite{KP99} of GSP with respect to a variety of equilibrium concepts, ranging from pure Nash and coarse-correlated equilibria in the full information setting to Bayes-Nash equilibria in the incomplete information setting. D{\"{u}}tting {\em et al.}~\cite{DFP14} proved bounds on the revenue and exploited more expressive input formats as a remedy for the non-existence of pure Nash equilibria in games induced by GFP. They designed the {\em expressive generalized first price auction} (EGFP) according to which each player submits a bid per position, the positions are auctioned off sequentially, and each player pays her bid for the position she is given. 

All of the aforementioned papers studied the {\em no-budget} setting, where the players are assumed to be able to afford any payments, no matter how large these can get. However, in reality, the players have hard budget constraints that upper-bound the payments that they can afford. Following a series of recent papers that focus on such budget-constrained settings, we also study the social efficiency of position mechanisms by bounding the (pure) price of anarchy and stability in terms of the {\em liquid welfare} benchmark that takes budgets into account. Liquid welfare was first introduced by Dobzinski and Paes Leme~\cite{DPL14} who focused on the design of truthful mechanisms for the allocation of multiple units of a single divisible item (see also \cite{LX15,LX17} for extensions of this setting). One of their very first results is the observation that the celebrated VCG mechanism~\cite{V61,C71,G73} is no longer truthful, which is the case for VCG in our setting as well. 

The liquid price of anarchy has been considered in a few related papers so far. Syrgkanis and Tardos~\cite{ST13} considered the liquid welfare benchmark under the term {\em effective welfare} and bounded the ratio between the optimal liquid welfare and the worst-case social welfare at equilibrium, in various strategic auction settings, including position mechanisms. Caragiannis and Voudouris~\cite{CV16} and Christodoulou {\em et al.}~\cite{CST16} were the first to provide constant bounds on the liquid price of anarchy (ratio of optimal liquid welfare over worst-case liquid welfare at equilibrium) of the proportional mechanism for the allocation of divisible resources. These results are based on the now-standard unilateral deviations technique (see also \cite{RST16}) and can be extended to more general equilibrium concepts, given a specific definition of the liquid welfare for randomized allocations. Our upper bounds follow this technique as well, but it seems non-trivial to extend them to more general equilibrium concepts due to the particular form of the deviating bids used. For pure equilibria in particular, by exploiting the structure of worst-case equilibria, Caragiannis and Voudouris~\cite{CV17} were able to characterize the liquid price of anarchy of all mechanisms for the allocation of a single divisible resource, leading to tight bounds. Finally, Azar {\em et al.}~\cite{AFGR15} refined the definition of the liquid welfare for randomized allocations and proved constant liquid price of anarchy bounds over general equilibrium concepts for simultaneous first price auctions. 

In Section~\ref{sec:prelim}, we formally describe the setting considered in this paper and the mechanisms that we are interested in. Then, in Section~\ref{sec:results} we prove our main result: the liquid price of anarchy and stability of GSP, VCG and EGFP is exactly $2$. Consequently, in contrast to the no-budget setting, when we consider players with budget constraints and the liquid welfare benchmark, these mechanisms do not have socially efficient equilibria. Such a phenomenon was first observed by Caragiannis and Voudouris~\cite{CV17} for all single divisible resource allocation mechanisms, and it might be the case that this holds for any position mechanism as well. We conclude with a short discussion of possible extensions of our work in Section~\ref{sec:open}. 

\section{Preliminaries}\label{sec:prelim}

There are $n$ available {\em positions} such that position $j$ has associated click-through-rate (CTR) $\alpha_j \in \R_{>0}$ such that $\alpha_j \geq \alpha_{j+1}$ for $j \in [n-1]$; let $\boldsymbol{\alpha} = (\alpha_j)_{j \in [n]}$ be the vector containing the CTRs of all positions. Furthermore, there are $n$ players that compete over these positions. Player $i$ has a valuation $v_i$ and a total private budget $c_i$; let $\vv = (v_i)_{i \in [n]}$ and $\cc = (c_i)_{i \in [n]}$ be the vectors containing the valuations and budgets of all players. The valuation $v_i$ indicates the value that player $i$ has per click and, therefore, if player $i$ is assigned to some position $j$, then her total value is $\alpha_j v_i$. The budget $c_i$ can be thought of as an upper bound to the payment that the player can afford in order to buy some position.

We consider several greedy mechanisms for the allocation of positions to players, which generally work as follows. Let $M$ be a greedy mechanism. Each player $i$ submits a bid $b_i$ that can either be a real non-negative scalar or a vector of scalars per position, depending on the input format that $M$ requires; let $\bb$ be the vector (or matrix) of bids submitted by all players. Then, the players are sorted in non-increasing order in terms of their bids and the induced ranking $\sigma(\bb)$ indicates the position $\sigma_i(\bb)$ that is assigned to each player $i$; therefore, we call $\sigma(\bb)$ an {\em assignment} that is induced by $\bb$. Also, let $\pi_j(\bb)$ denote the player that is assigned to position $j$ such that $\pi_{\sigma_i(\bb)}(\bb) = i$. The mechanism charges player $i$ an amount of money $p_i(\bb,M)$ that depends on $\bb$, and may or may not depend on $\alpha_{\sigma_i(\bb)}$. Given a bid vector $\bb$, each player $i$ has utility $u_i(\bb,M) = \alpha_{\sigma_i(\bb)}v_i - p_i(\bb,M)$ if $p_i(\bb,M) \leq c_i$, and $-\infty$ otherwise. We focus on three important greedy allocation mechanisms that function as follows.

\subsection*{Generalized Second Price (GSP)}
Each player $i$ submits a scalar $b_i \in \R_{\geq 0}$. The players are sorted in non-increasing order in terms of these bids and are assigned to the corresponding positions. Each player $i$ is charged the next highest bid per click, that is, the bid of player $\pi_{\sigma_i(\bb)+1}(\bb)$ who is assigned to the next position $\sigma_i(\bb)+1$. Hence, the payment of player $i$ is $p_i(\bb,\GSP) = \alpha_{\sigma_i(\bb)}b_{\pi_{\sigma_i(\bb)+1}(\bb)}$, and her utility can be written as $u_i(\bb,\GSP) = \alpha_{\sigma_i(\bb)} \left( v_i - b_{\pi_{\sigma_i(\bb)+1}(\bb)} \right)$. \footnote{Interestingly, D{\'{\i}}az {\em et al.}~\cite{DGKMS16} proved that GSP may not have any equilibria when the number of players exceeds the number of available positions and proposed alternative mechanisms; we here consider the same number of players and positions.}

\subsection*{Vickrey-Clarke-Groves (VCG)} 
Again, each player $i$ submits a scalar $b_i \in \R_{\geq 0}$, and the players are sorted in non-increasing order in terms of their bids. Each player $i$ is charged the difference between the social welfare (based on the bids) of the players ranked below $i$ if $i$ did not participate and their actual social welfare when $i$ participates. In other words, the payment of player $i$ is $p_i(\bb,\VCG) = \sum_{j=\sigma_i(\bb)+1}^n b_{\pi(j)}(\alpha_{j-1}-\alpha_j)$, and her utility can be written as $u_i(\bb,\VCG) = \alpha_{\sigma_i(\bb)} \left( v_i - \frac{1}{\alpha_{\sigma_i(\bb)}}\sum_{j=\sigma_i(\bb)+1}^n b_{\pi(j)}(\alpha_{j-1}-\alpha_j) \right)$.

\subsection*{Expressive Generalized First Price (EGFP)}
Each player $i$ submits a vector $\bb_i \in \R_{\geq 0}^n$ containing a bid per position. The positions are assigned to the players sequentially so that the next available position gets assigned to the player with the maximum bid for it, among the players that have not yet been allocated a position. In other words, let $S_j$ be the set of players that are competing for positions $\ell \geq j$; initially, $S_1$ contains all players. Then, $\pi_j(\bb,\EGFP) = \arg\max_{i \in S_j}{b_{i,j}}$. Each player $i$ is charged (in total) her bid for the position that she is allocated, i.e., $p_i(\bb,\EGFP) = b_{i,\sigma_i(\bb)}$, and her utility is $u_i(\bb,\EGFP) = \alpha_{\sigma_i(\bb)} v_i - b_{i,\sigma_i(\bb)}$.

\subsection*{The game} 
Let $M \in \{\GSP,\VCG, \EGFP\}$ be any of the aforementioned  position mechanisms. Mechanism $M$ induces a strategic position game $\calG(M)$ among the players who act as utility maximizers; this is true even for VCG as we will see in the next section. A bid vector (or matrix) $\bb$ is called a {\em pure Nash equilibrium} (or, simply, equilibrium) for $\calG(M)$ if all players simultaneously maximize their utilities and have no incentive to deviate to any different bid in order to increase their personal utility, i.e., $u_i(\bb,M) \geq u_i((y,\bb_{-i}),M)$, for all players $i$ and bids $y \neq b_i$. Here, the notation $(y,\bb_{-i})$ is used to denote the vector (or matrix) that is obtained by $\bb$ when player $i$ bids $y$ (and all other players bid according to $\bb$). Let $\eq(\calG(M))$ be the set of all equilibria of the position game $\calG(M)$.

\subsection*{Liquid welfare, price of anarchy and price of stability}
We measure the social efficiency of an assignment $\sigma(\bb)$ by the {\em liquid welfare} benchmark, which is defined as the total value of the players, with the value of each player capped by her budget, i.e., 
$$\LW(\sigma(\bb)) = \sum_i \min\{ \alpha_{\sigma_i(\bb)}v_i, c_i \}.$$
The {\em liquid price of anarchy} ({\em liquid price of stability}) of a position game $\calG(M)$ that is induced by a position mechanism $M$ is defined as the ratio between the optimal liquid welfare achieved by any assignment to the minimum (maximum) liquid welfare achieved at any equilibrium assignment. In other words, the liquid price of anarchy and the liquid price of stability of $\calG(M)$ are, respectively, equal to
$$\lpoa(\calG(M)) = \frac{\max_\yy \LW(\sigma(\yy))}{\min_{\bb \in \eq(\calG(M))}\LW(\sigma(\bb))}.$$
and
$$\lpos(\calG(M)) = \frac{\max_\yy \LW(\sigma(\yy))}{\max_{\bb \in \eq(\calG(M))}\LW(\sigma(\bb))}.$$
Then, the liquid price of anarchy and stability of a mechanism $M$ are respectively defined as the worst-case liquid price of anarchy and stability among all position games that are induced by $M$, i.e., $\lpoa(M) = \sup_{\calG(M)}\lpoa(\calG(M))$ and $\lpos(M) = \sup_{\calG(M)}\lpos(\calG(M))$ .

\subsection*{The no-over assumption: no-overbidding and no-overbudgeting}
For the GSP and VCG mechanisms, in order to have meaningful bounds on the liquid price of anarchy, we assume that $\alpha_{\sigma_i(\bb)}b_i \leq \min\{\alpha_{\sigma_i(\bb)}v_i,c_i\}$ for every player $i$. This is a combination of the well-known no-overbidding assumption that demands that $b_i \leq v_i$ and a no-overbudgeting assumption that demands that $\alpha_{\sigma_i(\bb)}b_i \leq c_i$. This assumption is necessary as it is easy (like in the case of the classic price of anarchy literature that deals with the social welfare objective) to construct position games that have arbitrarily bad liquid price of anarchy when the players overbid. For the EGFP mechanism such an assumption is of course not necessary due to the definition of the payment function.  

\section{Bounds on the liquid price of anarchy and stability}\label{sec:results}
We begin with Theorem~\ref{thm:lower-bounds}, where we show that the $\lpoa$ and $\lpos$ of GSP, VCG and EGFP are at least $2$; notice that the example that we present in the following proof also proves that VCG is no longer truthful when the players have budget constraints. Then, in Theorem~\ref{thm:upper-bounds} we prove that this bound of $2$ on the $\lpoa$ and $\lpos$ is tight. 

\begin{theorem}\label{thm:lower-bounds}
The liquid price of anarchy and stability of GSP, VCG (under the no-over assumption) and EGFP are at least $2$.
\end{theorem}

\begin{proof}
Let $M \in \{\GSP,\VCG,\EGFP\}$, $\lambda > 2$ and $\varepsilon \in (0,1/2)$. Consider a position game $\calG(M)$ among two players with valuations $\vv = (\lambda,1)$ and budgets $\cc = (1+\varepsilon,1)$, for two positions with CTRs $\boldsymbol{\alpha}= (1,1/\lambda)$. Observe that, for the two possible assignments $(1,2)$ and $(2,1)$, the liquid welfare is
$\LW(1,2) = \min\{ \lambda, 1+\varepsilon\} + \min\left\{ \frac{1}{\lambda}, 1 \right\} = \frac{(1+\varepsilon)\lambda+1}{\lambda}$
and
$\LW(2,1) = \min\{ 1, 1+\varepsilon\} + \min\{ 1, 1 \} = 2.$
Therefore, since $\lambda > 2$ and $\varepsilon < 1/2$, we have that $\LW(2,1) > \LW(1,2)$, and the optimal assignment is $(2,1)$. The ratio 
$$\frac{\LW(2,1)}{\LW(1,2)} = \frac{2\lambda}{(1+\varepsilon)\lambda+1}$$
tends to $2$ as $\lambda$ becomes arbitrarily large and $\varepsilon$ becomes arbitrarily small. In order to prove the theorem, it suffices to show that there exists an equilibrium bid vector that induces the assignment $(1,2)$, while there exists no equilibrium bid vector that induces the assignment $(2,1)$.

\paragraph{GSP.} First, consider the bid vector $\bb = (1+\varepsilon,1)$ which induces the assignment $(1,2)$. The utilities of the two players are $u_1(\bb,\GSP) = \lambda-1$ and $u_2(\bb,\GSP) = \frac{1}{\lambda}$. Player $2$ has no incentive to deviate as, by the no-over assumption, she cannot bid above her budget (which coincides with her value), while any other bid would not change the assignment. Player $1$ obviously has no incentive to deviate to any other bid $b_1' \geq b_2$ as the assignment as well as her payment would not change. So, consider the deviation of player $1$ to the bid $b_1'=b_2 - \gamma$, for some $\gamma>0$. Then, the induced assignment would be $(2,1)$ and player $1$ would have utility $u_1((b_1',b_2),\GSP) = 1 < u_1(\bb,\GSP)$ since $\lambda > 2$. Therefore, $\bb$ is an equilibrium, and the price of anarchy bound follows.
 
Now, assume that there exists an equilibrium bid vector $\bb=(b_1,b_2)$ with $b_1 \leq b_2 \leq 1$ so that the assignment $(2,1)$ is induced, while the no-over assumption is satisfied (for player $2$). The utilities of the two players at this equilibrium are $u_1(\bb,\GSP) = 1$ and $u_2(\bb,\GSP) = 1-b_1$. Consider the deviation of player $1$ to the bid $b_1' = c_1 = 1 + \varepsilon > b_2$. Then, the utility of this player would be $u_1((b_1',b_2),\GSP) = \lambda - b_2 \geq \lambda - 1 > 1$, since $b_2 \leq 1$ and $\lambda > 2$. Hence, player $1$ has incentive to deviate to $b_1'$ and $\bb$ cannot be an equilibrium. The price of stability bound follows.

\paragraph{VCG.} Like in the case of GSP, consider the bid vector $\bb = (1+\epsilon,1)$ which induces the assignment $(1,2)$. The payments of the players are $p_1(\bb,\VCG) = 1 - \frac{1}{\lambda}$ and $p_2(\bb,\VCG) = 0$, yielding utilities of $u_1(\bb,\VCG) = \lambda-1 + \frac{1}{\lambda}$ and $u_2(\bb,\VCG) = \frac{1}{\lambda}$. Obviously, again player $2$ has no incentive to deviate, while player $1$ has no incentive to deviate to any bid $b_1' \geq b_2$. So, consider the deviation of player $1$ to the bid $b_1'=b_2 - \gamma$, for some $\gamma>0$. Then, the induced assignment would be $(2,1)$, the payment of player $1$ would be $p_1((b_1',b_2),\VCG) = 0$ and her utility would be $u_1((b_1',b_2),\VCG) = 1 < \lambda-1 + \frac{1}{\lambda} = u_1(\bb,\VCG)$ since $(\lambda-1)^2 > 0$, for any $\lambda >2$. Therefore, $\bb$ is an equilibrium, and the price of anarchy bound follows.

Now, assume that there exists an equilibrium bid vector $\bb=(b_1,b_2)$ with $b_1 \leq b_2 \leq 1$ so that the assignment $(2,1)$ is induced, while the no-over assumption is satisfied (for player $2$). The payments of the players at this equilibrium are $p_1(\bb,\VCG) = 0$ and $p_2(\bb,\VCG) = b_1\left(1-\frac{1}{\lambda}\right)$, yielding utilities of $u_1(\bb,\VCG) = 1$ and $u_2(\bb,\VCG) = 1-b_1\left(1-\frac{1}{\lambda}\right)$. Consider the deviation of player $1$ to the bid $b_1' = c_1 = 1 + \varepsilon > b_2$. Then, the induced assignment would be $(1,2)$, while the payment and the utility of player $1$ would be $p_1((b_1',b_2),\VCG) = b_2\left(1-\frac{1}{\lambda}\right) \leq 1-\frac{1}{\lambda}$ and $u_1((b_1',b_2),\VCG) = \lambda - p_1(b_1',b_2) \geq \lambda - 1 + \frac{1}{\lambda} > 1$, respectively; the last inequality follows since $(\lambda -1)^2 > 0$, for any $\lambda > 2$. Hence, since player $1$ has incentive to deviate to $b_1'$, $\bb$ cannot be an equilibrium, and the price of stability bound follows.

\paragraph{EGFP.} To show that there exists an equilibrium bid vector $\bb$ that induces the assignment $(1,2)$, consider the bids $\bb_1 = (1+\delta,0)$, where $\delta > 0$ is arbitrarily small, and $\bb_2 = (1,0)$ of the two players for the two available positions, respectively. Observe that after the allocation of the first position, the second one is given without any competition to the only remaining player. Therefore, at equilibrium, no player has any incentive to submit a bid that is greater than zero for the second position. Player $2$ has no incentive to change her bid for the first position since she simply cannot bid any higher, while bidding any lower would not change the assignment. Player $1$ has no incentive to deviate to any other bid $b_{1,1}' \geq b_{2,1}$ as the allocation and her payment would not change, and $\delta$ is assumed to be arbitrarily small. So, consider the deviating bid $b_{1,1}' < b_{2,1}$ which would change the assignment to $(2,1)$ and the utility of player $1$ would be $u_1'(((b_{1,1}',0),\bb_2),\EGFP) = 1 < \lambda - 1 = u_1(\bb,\EGFP)$. Therefore, $\bb$ is indeed an equilibrium, and the price of anarchy bound follows.

For the price of stability bound, assume that there exists an equilibrium bid matrix $\bb$ such that  $b_{1,1} \leq b_{2,1} \leq 1$ so that the allocation $(2,1)$ is induced; again the two players must bid zero for the second position which is, basically, for free. The utilities of the two players at this equilibrium are $u_1(\bb,\EGFP) = 1$ and $u_2(\bb) = 1-b_{2,1}$. Consider the deviation of player $1$ to the bid $b_{1,1}' = 1 + \delta > b_{2,1}$ for the first position, where $\delta > 0$ is arbitrarily small. Then, player $1$ would be allocated the first position and her utility would be $u_1(((b_{1,1}',0),\bb_2),\EGFP) = \lambda - b_{1,1} = \lambda - 1 -\delta > 1$, since $\lambda > 2$ and $\delta$ is arbitrarily small. Hence, player $1$ has no incentive to deviate to $\bb_1'=(b_{1,1}',0)$, $\bb$ cannot be an equilibrium, and the proof is complete. 
\end{proof}
 
The proof of the upper bounds exploits the well-known technique (for proving welfare guarantees in games) of deviating bids. However, it is more complicated since the selected deviating bids must be such that the payments of the players are within their budgets. In fact, this is the main barrier in proving LPoA bounds for more general equilibrium concepts, like Bayes-Nash equilibria in the incomplete information model, where the bids of the players are random variables. 

\begin{theorem}\label{thm:upper-bounds}
The liquid price of anarchy and stability of GSP, VCG (under the no-over assumption) and EGFP are at most $2$.
\end{theorem}

\begin{proof}
Let $M \in \{\GSP,\VCG,\EGFP\}$ and consider any $n$-player position game $\calG(M)$ induced by $M$. Let $v_i$ and $c_i$ be the value and budget of player $i \in [n]$, and let $\alpha_j$ be the CTR of position $j \in [n]$. Let $\bb$ be an equilibrium bid vector that induces an assignment $\sigma(\bb)$; recall that $\pi_j(\bb)$ denotes the player that is assigned to position $j$. Moreover, let $o_i$ denote the position given to player $i$ at an optimal allocation, and $\OPT = \sum_i \min\{\alpha_{o_i}v_i,c_i\}$.

Now, consider the following partition of the players: $A = \{i: \alpha_{\sigma_i(\bb)}v_i \leq c_i\}$. Then, for every player $i \not\in A$, we have that $\min\{\alpha_{\sigma_i(\bb)}v_i,c_i\} = c_i \geq \min\{\alpha_{o_i}v_i,c_i\}$, and by summing over all such players, we obtain
\begin{align}\label{eq:sum-not-A}
\sum_{i \not\in A} \min\{\alpha_{\sigma_i(\bb)}v_i,c_i\} \geq \sum_{i \not\in A} \min\{\alpha_{o_i}v_i,c_i\}.
\end{align}
The rest of the proof is dedicated to showing that, for any player $i\in A$ and some $\gamma > 0$, it holds that 
\begin{align}\label{eq:to-be-proved}
u_i(\bb)\geq \min\{\alpha_{o_i}v_i,c_i\} -  \min\{\alpha_{o_i}v_{\pi_{o_i}(\bb)},c_{\pi_{o_i}(\bb)}\} - \gamma.
\end{align}
Then, since $\min\{\alpha_{\sigma_i(\bb)}v_i,c_i\} = \alpha_{\sigma_i(\bb)}v_i \geq u_i(\bb)$, by summing over all players $i \in A$, and by the fact that $|A| \leq n$, we obtain
\begin{align}\label{eq:sum-A}\nonumber
\sum_{i \in A} \min\{\alpha_{\sigma_i(\bb)}v_i,c_i\} &\geq \sum_{i \in A} \min\{\alpha_{o_i}v_i,c_i\} - \sum_{i \in A} \min\{\alpha_{o_i}v_{\pi_{o_i}(\bb)},c_{\pi_{o_i}(\bb)}\} - \gamma |A| \\
&\geq \sum_{i \in A} \min\{\alpha_{o_i}v_i,c_i\} - \LW(\sigma(\bb)) - \gamma n.
\end{align}
Hence, the theorem will follow by combining inequalities (\ref{eq:sum-not-A}) and (\ref{eq:sum-A}), and by choosing $\gamma$ to be arbitrarily small, since we have that
\begin{align*}
\LW(\sigma(\bb)) &= \sum_{i \not\in A} \min\{\alpha_{\sigma_i(\bb)}v_i,c_i\} + \sum_{i \in A} \min\{\alpha_{\sigma_i(\bb)}v_i,c_i\} \\
&\geq \sum_{i \not\in A} \min\{\alpha_{o_i}v_i,c_i\} +\sum_{i \in A} \min\{\alpha_{o_i}v_i,c_i\} - \LW(\sigma(\bb)) - \gamma n \\
&\geq \OPT - \LW(\sigma(\bb)) - \gamma n.
\end{align*}
We now distinguish between cases depending on which mechanism is used. In the following, since the mechanism under consideration is clear from context, we drop it from our notation.

\paragraph{GSP.} For any player $i \in A$ consider the deviating bid $b_{\pi_{o_i}(\bb)}+\tilde{\gamma}$, where $b_{\pi_{o_i}(\bb)}$ is the bid of the player that is given position $o_i$ at equilibrium and $\tilde{\gamma} = \frac{\gamma}{\alpha_{o_i}}$ is such that $b_{\pi_{o_i}(\bb)} \leq b_{\pi_{o_i}(\bb)}+\tilde{\gamma} \leq b_{\pi_{o_i - 1}(\bb)}$; notice that player $i$ can choose such a $\tilde{\gamma}$ as she has full information about the bids of the other players, and there exists a tie-breaking assigning the position $o_i$ to player $i$ after the deviation (in case of equality). With this deviating bid, player $i$ essentially plays only for her optimal position $o_i$, if she can afford to do so. 

If the deviating bid $b_{\pi_{o_i}(\bb)}+\tilde{\gamma}$ satisfies the no-over assumption, then player $i$ is guaranteed to be given position $o_i$ in the new allocation and pay $b_{\pi_{o_i}(\bb)}$ per click. By the equilibrium condition, the fact that $\gamma > 0$, and since $\alpha_{o_i}b_{\pi_{o_i}(\bb)} \leq \min\{\alpha_{o_i}v_{\pi_{o_i}(\bb)},c_{\pi_{o_i}(\bb)}\}$ (by the no-over assumption for player $\pi_{o_i}(\bb)$), we have that 
\begin{align*}
u_i(\bb) &\geq u_i(b_{\pi_{o_i}(\bb)}+\delta,\bb_{-i}) = \alpha_{o_i}(v_i - b_{\pi_{o_i}(\bb)})
\geq  \min\{\alpha_{o_i}v_i,c_i\} -  \min\{\alpha_{o_i}v_{\pi_{o_i}(\bb)},c_{\pi_{o_i}(\bb)}\} - \gamma.
\end{align*}

If the deviating bid does not satisfy the no-over assumption, then we have that $b_{\pi_{o_i}(\bb)}+\tilde{\gamma} > v_i$ or $\alpha_{o_i}(b_{\pi_{o_i}(\bb)}+\tilde{\gamma}) > c_i$. Due to the no-over assumption for player $\pi_{o_i}(\bb)$, both of these inequalities imply that $\min\{\alpha_{o_i}v_{\pi_{o_i}(\bb)},c_{\pi_{o_i}(\bb)}\} + \gamma > \min\{\alpha_{o_i}v_i,c_i\}$. Since player $i$ has non-negative utility at equilibrium, we conclude that 
\begin{align*}
u_i(\bb) \geq 0 > \min\{\alpha_{o_i}v_i,c_i\} - \min\{\alpha_{o_i}v_{\pi_{o_i}(\bb)},c_{\pi_{o_i}(\bb)}\} - \gamma,
\end{align*}
as desired.

\paragraph{VCG.} The proof is similar to that for GSP. The main difference here is that when the deviating bid $b_{\pi_{o_i}(\bb)}+\tilde{\gamma}$ of player $i \in A$ satisfies the no-over assumption, then player $i$ is again guaranteed to be given position $o_i$, but now has to pay $\sum_{j = o_i + 1}^n b_{\pi_j(\bb)}(\alpha_{j-1}-\alpha_j)$ in total. Observe that, since VCG is a greedy mechanism, at equilibrium we have that $b_{\pi_{o_i}(\bb)} \geq b_{\pi_j(\bb)}$ for every $j \in \{o_i+1, ..., n\}$. This implies that
\begin{align*}
\sum_{j = o_i + 1}^n b_{\pi_j(\bb)}(\alpha_{j-1}-\alpha_j) &\leq b_{\pi_{o_i}(\bb)} \sum_{j = o_i+1 }^{n}(\alpha_{j-1}-\alpha_j) 
=  b_{\pi_{o_i}(\bb)} (\alpha_{o_i}-\alpha_n) 
\leq \alpha_{o_i} b_{\pi_{o_i}(\bb)}.
\end{align*}
Using this, we can follow the proof template for GSP and show the desired inequality.

\paragraph{EGFP.} For any player $i\in A$ consider the deviating bid vector $\yy$ so that $y_{o_i} = b_{\pi_{o_i}(\bb),o_i}+\gamma$ and $y_j = 0$ for any other position $j \neq o_i$. Again, player $i$ plays only for her optimal position $o_i$, if she can afford to do so. If $y_{o_i} > c_i$, then since the utility of player $i$ is non-negative at equilibrium, we obtain
\begin{align*}
u_i(\bb) &\geq 0 > c_i - b_{\pi_{o_i}(\bb),o_i} - \gamma 
\geq  \min\{\alpha_{o_i}v_i,c_i\} -  \min\{\alpha_{o_i} v_{\pi_{o_i}(\bb)}, c_{\pi_{o_i}(\bb)} \} - \gamma,
\end{align*}
where the last inequality follows by the fact that player $\pi_{o_i}(\bb)$ has non-negative utility at equilibrium and her payment is within her budget, which imply that $b_{\pi_{o_i}(\bb),o_i} \leq \min\{\alpha_{o_i} v_{\pi_{o_i}(\bb)}, c_{\pi_{o_i}(\bb)} \}$.

Otherwise, the deviating bid is such that player $i$ is allocated position $o_i$ and her payment $y_{o_i}$ is within her budget. Therefore, by the equilibrium condition, and by the fact that $b_{\pi_{o_i}(\bb),o_i} \leq \min\{\alpha_{o_i} v_{\pi_{o_i}(\bb)}, c_{\pi_{o_i}(\bb)} \}$, we have that
\begin{align*}
u_i(\bb) &\geq u_i(\yy,\bb_{-i}) 
\geq \alpha_{o_i}v_i - b_{\pi_{o_i}(\bb),o_i} - \gamma 
\geq  \min\{\alpha_{o_i}v_i,c_i\} -  \min\{\alpha_{o_i}v_{\pi_{o_i}(\bb)},c_{\pi_{o_i}(\bb)}\} - \gamma
\end{align*}
and inequality (\ref{eq:to-be-proved}) follows. 
\end{proof}

\section{Possible extensions}\label{sec:open}
In this letter, we studied the efficiency of several well-known mechanisms for the allocation of (advertising) positions to strategic budget-constrained users, and proved that their liquid price of anarchy and stability for pure equilibria is exactly $2$. Of course, there are multiple interesting open questions that one could attempt to answer here, like exploring {\em all} position mechanisms and bounding their liquid price of anarchy and stability. In particular, is there any position mechanism with liquid price of anarchy strictly smaller than $2$, even for the fundamental case of two players? 

Another important direction for future research is to consider more general settings, with incomplete information where both the values and the budgets of the players are randomly drawn from some prior distribution, and bound the liquid price of anarchy of position mechanisms for more general equilibrium notions, like coarse-correlated and Bayes-Nash equilibria. Finally, it might be interesting to study scenarios where the budgets of the players are assumed to be common knowledge (or they can be inferred in some way), and design mechanisms with improved social efficiency guarantees, by exploiting this information.

\vskip 0.2in
\bibliographystyle{plain}
\bibliography{gsp}

\end{document}